\documentclass{article}
\usepackage{xcolor,svgcolor}
\usepackage{hyperref}
\makeatletter
\hypersetup{
  pdftitle={\@title},
  pdfauthor={\@author},
  breaklinks=true,
  colorlinks=true,
  linkcolor=blue,
  citecolor=darkred,
  urlcolor=darkblue, 
}
\makeatother

\usepackage{amsmath,amssymb,amsthm}

\usepackage{url,xspace}
\usepackage{bussproofs}
\usepackage{mathbbol}
\usepackage[all]{xy}
\SelectTips{cm}{}

\usepackage{stmaryrd}

\usepackage{tikz}

\theoremstyle{plain}
\newtheorem{theorem}{Theorem}[section]
\newtheorem{proposition}[theorem]{Proposition}
\newtheorem{lemma}[theorem]{Lemma}
\newtheorem{corollary}[theorem]{Corollary}

\theoremstyle{definition}
\newtheorem{definition}[theorem]{Definition}

\newcommand{\ntbox}[1]{\,\,\makebox[0pt][c]{{$\Box$}}\makebox[0pt][c]{\raisebox{2pt}{\tiny\ensuremath{#1}}}\,\,}

\newcommand{\pure}{{(0)}} 
\newcommand{\cst}{{(1)}} 
\newcommand{\acc}{{(1)}} 
\newcommand{\modi}{{(2)}} 
\newcommand{\dec}{{(d)}} 
 
\newcommand{\decp}{{(d')}}

\newcommand{\eqbox}{\ntbox{\equiv}}

\newcommand{\eqs}{\cong} 
\newcommand{\eqw}{\sim}
 
\newcommand{\id}{\mathit{id}}
\newcommand{\cC}{\mathbf{C}}

\newcommand{\pr}{\mathit{pr}} 
\newcommand{\copr}{\mathit{in}} 
\newcommand{\pair}[1]{\langle #1 \rangle} 
\newcommand{\lpair}[1]{\langle #1 \rangle_l} 
\newcommand{\rpair}[1]{\langle #1 \rangle_r} 
\newcommand{\pa}{\langle \; \rangle} 
\newcommand{\copair}[1]{[ #1 ]} 
\newcommand{\lcopair}[1]{[ #1 ]_l} 
 
\newcommand{\copa}{[\;]} 
\newcommand{\eps}{\varepsilon} 
\newcommand{\To}{\Rightarrow}

\newcommand{\eqn}{\mathit{eq}}
\newcommand{\mon}{\mathit{mon}}
\newcommand{\exc}{\mathit{exc}}
\newcommand{\comon}{\mathit{com}}
\newcommand{\sta}{\mathit{st}}
\newcommand{\downcast}{{\downarrow}} 
\newcommand{\empt}{\mathbb{0}}
\newcommand{\unit}{\mathbb{1}}
 
\newcommand{\tagg}{\mathtt{tag}}
\newcommand{\untag}{\mathtt{untag}}
\newcommand{\throw}{\mathtt{throw}}
\newcommand{\try}{\mathtt{try}}
\newcommand{\catch}{\mathtt{catch}}
\newcommand{\TRY}{\mathtt{TRY}}
\newcommand{\lookup}{\mathtt{lookup}}
\newcommand{\update}{\mathtt{update}}
\newcommand{\squad}{\;\;}
\newcommand{\stimes}{\!\times\!}
\newcommand{\splus}{\!+\!}
\newcommand{\scolon}{\!\colon\!}
\newcommand{\sto}{\!\to\!}
\newcommand{\Ename}{\mathit{Exn}}
\newcommand{\Loc}{\mathit{Loc}}
\newcommand{\abr}{\boxed{\mathit{abrupt}}}
\newcommand{\nor}{\boxed{\mathit{normal}}}
\newcommand{\M}{M} 
\newcommand{\coM}{T} 
\newcommand{\isexc}{\txt{exc?}} 
 
\newcommand{\Tt}{\mathcal{T}} 
\newcommand{\Log}{\mathcal{L}} 

\newcommand{\rul}[1]{\textrm{(#1)}} 

\title{Patterns for computational effects arising from a monad or a comonad}
\author{Jean-Guillaume Dumas\thanks{Laboratoire J. Kuntzmann,
    Universit\'e de Grenoble. 51, rue des Math\'ematiques, umr CNRS
    5224, bp 53X, F38041 Grenoble, France,
    \href{mailto:Jean-Guillaume.Dumas@imag.fr,Dominique.Dumas@imag.fr}{\{Jean-Guillaume.Dumas,Dominique.Duval\}@imag.fr}.}
  \and Dominique Duval\footnotemark[1]
  \and Jean-Claude Reynaud\thanks{Reynaud Consulting (RC), \href{mailto:Jean-Claude.Reynaud@imag.fr}{Jean-Claude.Reynaud@imag.fr}.}
}

\begin{document}

\maketitle

\begin{abstract}
This paper presents equational-based logics for proving first order 
properties of programming languages involving effects. 

We propose two dual inference system patterns 
that can be instanciated with monads or comonads 
in order to be used for proving properties of different effects.
The first pattern provides inference rules which can be interpreted 
in the Kleisli category of a monad and the coKleisli category 
of the associated comonad. In a dual way, the second pattern 
provides inference rules which can be interpreted in the 
coKleisli category of a comonad and the Kleisli category of the 
associated monad. The logics combine a 3-tier effect system 
for terms consisting of pure terms and two other kinds of effects 
called 'constructors/observers' and 'modifiers', 
and a 2-tier system for 'up-to-effects' and 'strong' equations. 
Each pattern provides generic rules for dealing with any monad 
(respectively comonad), and it can be extended with specific rules 
for each effect. The paper presents two use cases: 
a language with exceptions (using the standard monadic semantics), 
and a language with state (using the less standard comonadic semantics). 
Finally, we prove that the obtained inference system for states 
is Hilbert-Post complete.

\end{abstract}

\section{Introduction} 

A \emph{software design pattern} is not a finished design,
it is a description or template that can 
be instanciated in order to be used in many different situations. 
In this paper, we propose \emph{inference system patterns}  
that can be instanciated with monads or comonads 
in order to be used for proving properties of different effects.

In order to formalize computational effects one can choose 
between types and effects systems \cite{LucassenGifford88}, 
monads \cite{Moggi91} and their associated Lawvere theories 
\cite{PlotkinPower02}, comonads \cite{UustaluVene08}, 
or decorated logics \cite{DD10-dialog}. 
Starting with Moggi's seminal paper \cite{Moggi91} 
and its application to Haskell \cite{Wa92}, 
various papers deal with the effects arising from a monad,
for instance \cite{PlotkinPower02,SM04,Le06,PP09}.

Each of these approaches rely on some classification of the syntactic 
expressions according to their interaction with effects. 
In this paper we use decorated logics which, by extending this classification 
to equations, provide a proof system adapted to each effect. 

This paper presents equational-based logics for proving first order 
properties of programming languages involving effects. 
We propose two dual patterns, consisting in a language with 
an inference system, for building such a logic. 

The first pattern provides inference rules which can be interpreted 
in the coKleisli category of a comonad and the Kleisli category 
of the associated monad. In a dual way, the second pattern 
provides inference rules which can be interpreted in the 
Kleisli category of a monad and the coKleisli category of the 
associated comonad. The logics combine a three-levels effect system 
for terms consisting of pure terms and two other kinds of effects 
called observers/constructors and modifiers, 
and a two-levels system for strong and weak equations. 

Each pattern provides generic rules for dealing with any comonad 
(respectively monad), and it can be extended with specific rules 
for each effect. The paper presents two use cases:  
a language with state and a language with exceptions. 
For the language with state we use a comonadic semantics 
and we prove that the equational theory obtained is Hilbert-Post complete,
which provides a new proof for a result in \cite{Pretnar10}. 
For the language with exceptions we extend the standard monadic semantics 
in order to catch exceptions; this relies on the duality between states 
and exceptions from~\cite{DDFR12-dual}.  

We do not claim that each effect arises either from a comonad or from
a monad, but this paper only deals with such effects.  
Intuitively, an effect which observes features may arise from a comonad, 
while an effect which constructs features may arise from a monad
\cite{JacobsRutten11}.
However, some interesting features in the comonad pattern stem from the 
well-known fact 
that each comonad determines a monad on its coKleisli category, 
and dually for the monad pattern. 
More precisely, on the monads side, 
let $(\M,\eta,\mu)$ be a monad on a category $\cC^\pure$ 
and let $\cC^\cst$ be the Kleisli category of $(\M,\eta,\mu)$ on $\cC^\pure$.
Then $\M$ can be seen as the endofunctor of a comonad $(\M,\eps,\delta)$ 
on $\cC^\cst$, so that we may consider the coKleisli category $\cC^\modi$ 
of $(\M,\eps,\delta)$ on $\cC^\cst$.
The canonical functors from $\cC^\pure$ to $\cC^\cst$ 
and from $\cC^\cst$ to $\cC^\modi$ give rise to a hierarchy of terms: 
pure terms in $\cC^\pure$, constructors in $\cC^\cst$, modifiers in $\cC^\modi$. 
This corresponds to the three translations
of a typed lambda calculus into a monadic language \cite{Wa92}.

On the comonads side, we get a dual hierarchy: 
pure terms in $\cC^\pure$, observers in $\cC^\cst$, modifiers in $\cC^\modi$. 

We instanciate these patterns with  
two fundamental examples of effects: state and exceptions. 

Following \cite{DDFR12-dual}, we consider that 
the states effect arise from the comonad $A\times S$ 
(where $S$ is the set of states), 
thus a decorated logic for states is built by extending 
the pattern for comonads.
The comonad itself provides a decoration for the lookup operation, 
which observes the state, 
while the monad on its coKleisli category 
provides a decoration for the update operation.

Following \cite{Moggi91}, we consider that 
the exceptions effect arise from the monad $A+E$ 
(where $E$ is the set of exceptions),
thus a decorated logic for exceptions is built by extending 
the pattern for monads.
The monad itself provides a decoration for the raising operation, 
which constructs an exception, 
while the comonad on its Kleisli category 
provides a decoration for the handling operation.

In fact the decorated logic for exceptions is not exactly dual to 
the decorated logic for states if we assume that 
the intended interpretation takes place in a  
distributive category, like the category of sets, 
which is not codistributive. 

Other effects would lead to other additional rules, but we 
have chosen to focus on two effects which are well known 
from various points of view. 
Our goal is to enligthen the contributions of each approach: 
the annotation system from the types and effects 
systems \cite{LucassenGifford88}, 
the major role of monads for some effects \cite{Moggi91},
and the dual role of comonads \cite{UustaluVene08}, 
as well as the flexibility of decorated logics \cite{DD10-dialog}. 
Moreover, proofs in decorated logics can be checked 
with the Coq proof assistant; a library for states is available there:
\url{http://coqeffects.forge.imag.fr}.

In this paper we focus on finite products and coproducts; 
from a programming point of view this means that we are 
considering languages with $n$-ary operations and with case distinction,
but without loops or higher-order functions. 
In a language with effects there is a well-known issue with 
$n$-ary operations: their interpretation may depend on the order 
of evaluation of their arguments. In this paper 
we are looking for languages with case distinction 
and with \emph{sequential products}, 
which allows to force the order of evaluation of the arguments, 
whenever this is required. 

It is well known that (co)monads fit very well with composition 
but require additional assumptions for being 
fully compatible with products and coproducts.
This corresponds to the fact that 
in the patterns from Section~\ref{sec:patterns},
which are valid for any (co)monad, 
the rules for products and coproducts 
hold only under some decoration constraints. 
However, such assumptions are satisfied for several (co)monads.
This is in particular the case for the state comonad and the 
exceptions monad.

In Section~\ref{sec:patterns} we describe the patterns for 
a comonad and for a monad. 
The first pattern is instanciated with the comonad for state 
in Section~\ref{sec:states},
and we prove the Hilbert-Post completeness of the decorated theory for state. 
In Section~\ref{sec:exceptions} we instanciate the second pattern 
to the monad for exceptions.

\section{Patterns for comonads and for monads } 
\label{sec:patterns} 

\subsection{Equational logic with conditionals}\

In this Section we define 
a grammar and an inference system for two logics $\Log_\comon$ and $\Log_\mon$,
then we define an interpretation of these logics in 
a category with a comonad and a monad, respectively. 

The logics $\Log_\comon$ and $\Log_\mon$ are called \emph{decorated} logics 
because their grammar and inference rules are essentially the 
grammar and inference rules for a ``usual'' logic,
namely the equational logic with conditionals (denoted $\Log_\eqn$), 
together with \emph{decorations} for the terms and for the equations.
The decorations for the terms are similar to the \emph{annotations} 
of the types and effects systems \cite{LucassenGifford88}. 

Decorated logics are introduced in \cite{DD10-dialog} 
in an abstract categorical framework, which will not be explicitly used 
in this paper.  

The grammar of the equational logic with conditionals 
is reminded in Figure~\ref{fig:logeq-grammar}. 
Each term has a source type and a target type.
As usual in categorical presentations of equational logic, 
a term has precisely one source type,  
which can be a product type or the unit type. 
Each equation relates two parallel terms, i.e.,
two terms with the same source and the same target. 
This grammar will be extended with decorations 

in order to get the grammar of the logics $\Log_\comon$ and $\Log_\mon$.

\begin{figure}[htbp]   
\renewcommand{\arraystretch}{1.3}
$$ \begin{array}{|ll|} 
\hline
\multicolumn{2}{|c|}{
\mbox{ Grammar for the equational logic with conditionals: }} \\
\hline
\textrm{Types: } & t::= 
   A\mid B\mid \dots\mid t+t\mid\empt\mid  t\times t\mid \unit  \\ 
\textrm{Terms: } & f::=   \id_t\mid f\circ f\mid  \\
& \qquad \pair{f,f} \mid \pr_{t,t,1}\mid \pr_{t,t,2}\mid \pa_t  \\ 
& \qquad \copair{f|f}\mid \copr_{t,t,1}\mid\copr_{t,t,2}\mid\copa_t\mid  \\
 \textrm{Equations: } & e::=  f\equiv f   \\  
\hline
\end{array}$$
\renewcommand{\arraystretch}{1}
\caption{Equational logic with conditionals: grammar} 
\label{fig:logeq-grammar} 
\end{figure}

\subsection{Patterns}\

The rules in Figure~\ref{fig:pattern-rules} 
are \emph{patterns}, in the following sense:
when the boxes in the rules are removed, 
we get usual rules for the logic $\Log_\eqn$,
which may be interpreted in any bicartesian category.
When the boxes are replaced by decorations, 
we get a logic which, according to the choice of decorations, 
may be interpreted in a bicartesian category with a comonad or a monad.
There may be other ways to decorate the rules for $\Log_\eqn$, 
but this is beyond the scope of this paper.

\begin{figure}[htbp]   
\renewcommand{\arraystretch}{2}
$$ \begin{array}{|l|} 
\hline
\mbox{congruence rules} \\ 
\rul{refl} \quad 
  \dfrac{f^\Box}{f \eqbox f} \qquad
\rul{sym} \quad 
  \dfrac{f^\Box \eqbox g^\Box}{g \eqbox f}  \qquad
\rul{trans} \quad 
  \dfrac{f^\Box \eqbox g^\Box \squad g^\Box \eqbox h^\Box}{f \eqbox h}  \\ 
\rul{repl} \quad 
  \dfrac{f_1^\Box\eqbox f_2^\Box\colon A\to B \squad g^\Box\colon B\to C}
    {g\circ f_1 \eqbox g\circ f_2 }  \qquad
\rul{subs} \quad 
  \dfrac{f^\Box\colon A\to B \squad g_1^\Box\eqbox g_2^\Box\colon B\to C}
    {g_1 \circ f \eqbox g_2\circ f } \\
\hline
\mbox{categorical rules} \\ 
\rul{id} \quad
  \dfrac{A}{\id_A^\Box\colon A\to A } \qquad 
\rul{comp} \quad
  \dfrac{f^\Box\colon A\to B \quad g^\Box\colon B\to C}
    {(g\circ f)^\Box \colon A\to C}  \\
\rul{id-source} \quad 
  \dfrac{f^\Box\colon A\to B}{f\circ \id_A \eqbox f} \qquad 
\rul{id-target} \quad 
  \dfrac{f^\Box\colon A\to B}{\id_B\circ f \eqbox f} \\
\rul{assoc} \quad 
  \dfrac{f^\Box\colon A\to B \squad g^\Box\colon B\to C \squad h^\Box\colon C\to D}
  {h\circ (g\circ f) \eqbox (h\circ g)\circ f}  \\
\hline
\mbox{product rules} \\ 
\rul{prod} \quad
  \dfrac{B_1 \quad B_2 }
    {B_1\stimes B_2 \quad \pr_1^\Box\colon B_1\stimes B_2 \to B_1 \quad 
    \pr_2^\Box\colon B_1\stimes B_2 \to B_2}  \\ 
\rul{pair}  \quad
  \dfrac{ f_1^\Box\colon  A \to B_1 \quad f_2^\Box\colon  A \to B_2}
    {\pair{f_1,f_2}^\Box\colon  A\to B_1\stimes B_2} \\
\rul{pair-eq} \quad 
  \dfrac{f_1^\Box\colon A\to B_1 \squad f_2^\Box\colon A\to B_2}
    {\pr_1\circ\pair{f_1,f_2} \eqbox f_1 \quad
  \pr_2\circ\pair{f_1,f_2} \eqbox f_2 }  \\
\rul{pair-u} \quad 
  \dfrac{f_1^\Box\scolon A\!\sto\! B_1 \squad
    f_2^\Box\scolon A\!\sto\! B_2 \squad g^\Box\scolon A\!\sto\! B_1\stimes B_2 
    \squad \pr_1\circ g\eqbox f_1 \squad \pr_2\circ g\eqbox f_2 }
    {g \eqbox \pair{f_1,f_2}} \\
\rul{unit} \quad
  \dfrac{}{\unit} \qquad
\rul{final} \quad
  \dfrac{A}{\pa_A^\Box\colon A\to \unit} \qquad
\rul{final-u} \quad 
  \dfrac{f^\Box\colon A\to \unit}{f \eqbox \pa_A} \\
\hline
\mbox{coproduct rules} \\ 
\rul{coprod} \quad
  \dfrac{A_1 \quad A_2}
    {A_1\splus A_2 \quad \copr_1^\Box\colon A_1\to A_1\splus A_2 \quad 
    \copr_2^\Box\colon A_2\to A_1\splus A_2 } \\ 
\rul{copair} \quad 
  \dfrac{ f_1^\Box\colon  A_1 \to B \quad f_2^\Box\colon  A_2 \to B}
    {\copair{f_1|f_2}^\Box\colon  A_1\splus A_2 \to B } \\
\rul{copair-eq} \quad 
  \dfrac{ f_1^\Box\colon  A_1 \to B \squad f_2^\Box\colon  A_2 \to B}
    {\copair{f_1|f_2} \circ \copr_1 \eqbox  f_1 \quad
    \copair{f_1|f_2} \circ \copr_2 \eqbox  f_2  } \\
\rul{copair-u} \quad 
  \dfrac{g^\Box\scolon A_1\!\splus\! A_2 \sto B \squad 
    f_1^\Box\scolon A_1 \sto B 
    \squad f_2^\Box\scolon A_2 \sto B \squad 
    g\circ \copr_1 \eqbox f_1 \squad g\circ \copr_2 \eqbox f_2 }
    {g \eqbox \copair{f_1|f_2}}  \\
\rul{empty} \quad
  \dfrac{}{\empt} \qquad
\rul{initial} \quad
  \dfrac{B}{\copa_B^\Box\colon \empt\to B} \qquad
\rul{initial-u} \quad 
  \dfrac{g^\Box\colon \empt\to B}{g \eqbox \copa_B} \\
\hline 
\end{array}$$
\renewcommand{\arraystretch}{1}
\caption{Patterns: rules} 
\label{fig:pattern-rules} 
\end{figure}

\subsection{A decorated logic for a comonad}\

In the logic $\Log_\comon$ for comonads, 
each term has a decoration which is denoted as a superscript
$\pure$, $\acc$ or $\modi$:
a term is \emph{pure} when its decoration is $\pure$,
it is an \emph{accessor} (or an \emph{observer}) when its decoration is $\acc$
and a \emph{modifier} when its decoration is $\modi$. 
Each equation has a decoration which is denoted by 
replacing the symbol $\equiv$ either by $\eqs$ or by $\eqw$:
an equation with $\eqs$ is called \emph{strong},
with $\eqw$ it is called \emph{weak}.

The inference rules of $\Log_\comon$ are obtained by 
introducing some \emph{conversion} rules 
and by decorating the rules in Figure~\ref{fig:pattern-rules}.

When writing terms, if a decoration does not matter  
or if it is clear from the context, it may be omitted.  
\begin{itemize}
\item 
The conversion rules are:
$$ \begin{array}{|c|}
\hline
\dfrac{f^\pure}{f^\acc} \qquad \dfrac{f^\acc}{f^\modi} 
  \qquad \dfrac{f^\dec\eqs g^\decp}{f\eqw g} \mbox{ for all } d,d' \qquad
\dfrac{f^\dec\eqw g^\decp}{f\eqs g} \mbox{ for all } d,d'\leq 1 \\
\hline
\end{array}$$
The conversions for terms are \emph{upcasting} conversions.

We will always use them in a \emph{safe} way, by 
interpreting them as injections. 
This allows to avoid any specific notation for these conversions; 
an accessor $a^\acc$ may be converted to a modifier 
which is denoted $a^\modi$:  
both have the same name although they are distinct terms;
similarly, a pure term $v^\pure$ may be converted to $v^\acc$ or to $v^\modi$.
An equation between terms with distinct decorations 
does not imply any downcasting of its members; 
for instance, if $f^\modi\eqs g^\pure$ then it does not follow 
that $f$ is downcasted to $f^\pure$. 
The conversions for equations mean that strong and weak equations
coincide on pure terms and accessors and that 
each strong equation between modifiers can be seen as a weak one.
\item 
All rules of $\Log_\eqn$ are decorated with $\pure$ for terms and $\eqs$ for 
equations: the pure terms with the strong equations form 
a sublogic of $\Log_\comon$ which is isomorphic to $\Log_\eqn$. 
Thus we get $\id^\pure$, $\pr^\pure$, $\pa^\pure$, $\copr^\pure$, $\copa^\pure$.
\item 
The congruence rules for equations take all decorations for terms 
and for equations, with one notable exception:
the replacement rule for weak equations 
holds only when the replaced term is pure:
$$ \begin{array}{|c|}
\hline
  \rul{repl} \quad 
  \dfrac{f_1^\dec\eqw f_2^\decp\colon A\to B \squad g^\pure\colon B\to C}
    {g\circ f_1 \eqw g\circ f_2 } \\
\hline
\end{array}$$
\item 
The categorical rules hold for all decorations and
the decoration of a composed terms is 
the maximum of the decorations of its components. 
\item 
The product rules hold only when the given terms are pure or accessors and
the decoration of a pair is 
the maximum of the decorations of its components. 
Thus, $n$-ary operations can be used only when their arguments are 
accessors. 
 \item 
The coproduct rules hold only when the given terms are pure and
a copair is always pure, which is 
the maximum of the decorations of its components.
Thus, case distinction can be done only for pure terms.
\end{itemize}

\subsection{The interpretation of \texorpdfstring{$\Log_\comon$}{Lcom} by a comonad}\

In order to give a meaning to the logic $\Log_\comon$, 
let us consider a bicartesian category $\cC$ with 
a comonad $(\coM,\eps,\delta)$ satisfying the epi requirement,
i.e., $\eps_A:\coM A \to A$ is an epimorphism for each object $A$
(the dual assumption is discussed in \cite{Moggi91}). 

Then we get a model $\cC_\coM$ of the decorated logic $\Log_\comon$
as follows.
\begin{itemize}
\item The types are interpreted as the objects of $\cC$.
\item The terms are interpreted as morphisms of $\cC$: 
a pure term $f^\pure\colon A\!\to\! B$ as a morphism 
$f\colon A\!\to\! B$ in~$\cC$;
an accessor $f^\acc\colon\! A\!\to\! B$ as a morphism 
$f\colon\! \coM A\!\to\! B$ in~$\cC$;
and a modifier $f^\modi\colon A\!\to\! B$ as a morphism 
$f\colon \coM A\!\to\! \coM B$ in~$\cC$.
\item The conversion from pure terms to accessors 
is interpreted by mapping $f\colon A\to B$ to 
$f\circ \eps_A \colon  \coM A\to B$.
The epi requirement implies that this conversion is safe.
\item The conversion from accessors to modifiers
is interpreted by mapping $f\colon \coM A\to B$ to 
$\coM f\circ\delta_A \colon  \coM A\to \coM B$. 
It is easy to check that this conversion is safe.

\item When a term $f$ 
has several decorations (because it is pure or accessor, 
and thus can be upcasted) we will denote by $f$ 
any one of its interpretations:
a pure term $f^\pure\colon A\!\to\! B$ may be interpreted 
as $f\colon A\!\to\! B$ and as 
$f\colon \coM A\!\to\! B$ and as 
$f\colon \coM A\!\to\! \coM B$,
and an accessor $f^\acc\colon\! A\!\to\! B$ as 
$f\colon \coM A\!\to\! B$ and as 
$f\colon \coM A\!\to\! \coM B$. 
The choice will be clear from the context, and when several choices 
are possible they will give the same result, up to conversions. 
For this reason, we will describe the interpretation of the rules 
only for the largest possible decorations. 
\item The identity $\id_A^\pure\colon A\to A$ is interpreted as
$\id_A\colon A\to A$ in $\cC$;
\item The composition of two modifiers $f^\modi\colon A\to B$ and 
$g^\modi\colon B\to C$ is interpreted as $g\circ f\colon \coM A\to \coM B$ 
in~$\cC$. 
\item An equation between modifiers $f^\modi\eqs g^\modi\colon A\to B$ 
is interpreted by an equality $f=g\colon \coM A\to \coM B$ in~$\cC$. 
\item A weak equation between modifiers $f^\modi\eqw g^\modi\colon A\to B$ 
is interpreted by an equality 
$\eps_B\circ f=\eps_B\circ g\colon \coM A\to B$ in $\cC$. 
\item The unit type is interpreted as the final object of $\cC$
and the term $\pa_A^\pure\colon A\to\unit$ as the unique morphism 
from $A$ to $\unit$ in~$\cC$.
\item The product $B_1\times B_2$ with its projections 
is interpreted as the binary product in~$\cC$
and the pair of $f_1^\pure\colon A\to B_1$ and $f_2^\pure\colon A\to B_2$ 
as the pair $\pair{f_1,f_2}\colon A\to B_1\times B_2$ in~$\cC$. 
\item The empty type is interpreted as the initial object of $\cC$
and the term $\copa_A^\pure\colon \empt\to A$ as the unique morphism 
from $\empt$ to $A$ in~$\cC$.
\item The coproduct $A_1+A_2$ with its coprojections 
is interpreted as the binary coproduct in $\cC$ and 
the copair of $f_1^\acc\colon A_1\to B$ and $f_2^\acc\colon A_2\to B$ 
as the copair $\copair{f_1|f_2} \colon A_1+A_2\to \coM B$ in~$\cC$. 
\end{itemize}

\subsection{A decorated logic for a monad}\

The dual of the decorated logic $\Log_\comon$ for a comonad 
is the decorated logic $\Log_\mon$ for a monad. 

Thus, the grammar of $\Log_\mon$ is the same as the grammar of $\Log_\comon$,
but a term with decoration $\cst$ is now called a \emph{constructor}.

The rules for $\Log_\mon$ are nearly the same 
as the corresponding rules for $\Log_\comon$, except that
for weak equations the replacement rule always holds 
while the substitution rule holds only when the substituted term is pure:
$$ \begin{array}{|c|}
\hline
  \rul{subs} \quad 
  \dfrac{f^\pure\colon A\to B \squad g_1^\dec\eqw g_2^\decp\colon B\to C}
    {g_1 \circ f \eqw g_2\circ f } \\
\hline
\end{array}$$
In the rules for pairs and copairs, the decorations are permuted.  

The logic $\Log_\mon$ can be interpreted dually to $\Log_\comon$.
Let $\cC$ be a bicartesian category 
and $(\M,\eta,\mu)$ a monad on $\cC$ satisfying the mono requirement,
which means that $\eta_A:A\to \M A$ is a monomorphism for each object $A$. 
Then we get a model $\cC_{\M}$ of the decorated logic $\Log_\mon$, where 

a constructor $f^\cst\colon\! A\!\to\! B$ 
is interpreted as a morphism $f\colon A\to \M B$ in~$\cC$ 

and a weak equation $f^\modi\eqw g^\modi\colon A\to B$ is interpreted 
as an equality 
$f\circ\eta_A=g\circ\eta_A\colon A\to \coM B$ in $\cC$.

\section{States: an instance of the pattern for comonads}
\label{sec:states}

\subsection{A decorated logic for state}\

Let us consider a distributive category $\cC$ 
with epimorphic projections 
and with a distinguished object $S$ called the \emph{object of states}.
We consider the comonad $(\coM,\eps,\delta)$  
with endofunctor $\coM A=A\times S$, 
with counit $\eps$ made of the projections $\eps_A\colon A\times S\to A$, and 
with comultiplication $\delta$ which ``duplicates'' the states, 
in the sense that $\delta_A=\pair{\id_{A\times S}|\pr_A}\colon 
A\times S \to (A\times S)\times S$ where 
$\pr_A\colon A\times S\to A$ is the projection. 

We call this comonad the \emph{comonad of state}.
It is sometimes called the \emph{product comonad}, 
and it is different from the \emph{costate comonad} or
\emph{store comonad} with endofuntor $ \coM A=S\times A^S $
\cite{GJ12}.

The category $\cC$ with the comonad of states 
provides a model of the logic $\Log_\comon$. 
We can extend $\Log_\comon$ into a logic $\Log_\sta$  
dedicated to the state comonad. 

First, because of the specific choice of the comonad $\coM A=A\times S$,
we can add new decorations to the rule patterns for pairs in $\Log_\comon$, 
involving modifiers: 
there is a \emph{left pair} $\lpair{f_1,f_2}^\modi$ 
of an accessor $f_1^\acc$ and a modifier $f_2^\modi$,
satisfying the first three rules in Figure~\ref{fig:state-prod}. 
There are also three rules (omitted), symmetric to these ones, for the 
\emph{right pair} $\rpair{f_1,f_2}^\modi$ 
of a modifier $f_1^\modi$ and an accessor $f_2^\acc$.

The interpretation of the left pair $\lpair{f_1,f_2}^\modi:
A\to B_1\times B_2$ is the pair 
$\pair{f_1,f_2}:A\times S \to B_1\times B_2\times S$ 
of $f_1:A\times S\to B_1$ and $f_2:A\times S \to B_2\times S$.

Moreover, the rule \rul{effect} expresses the fact that,
when $\coM A=A\times S$, two modifiers coincide as soon as 
they return the same result and modify the state in the same way.

\begin{figure}[htbp]   
\renewcommand{\arraystretch}{2}
$$ \begin{array}{|ll|} 
\hline
\rul{l-pair} &
\dfrac{ f_1^\cst\colon A\to B_1 \quad f_2^\modi\colon  A \to B_2}
  {\lcopair{f_1|f_2}^\modi\colon  A \to B_1\times B_2 }  \\
\rul{l-pair-eq} &
\dfrac{ f_1^\cst\colon A\to B_1 \quad f_2^\modi\colon  A \to B_2}
  {\pr_1^\pure \circ \lpair{f_1,f_2}^\modi \eqw f_1^\cst \quad 
   \pr_2^\pure \circ \lpair{f_1,f_2}^\modi \eqs f_2^\modi  } \\
\rul{l-pair-u} &
\dfrac{g^\modi\scolon A \sto B_1\stimes B_2 \squad 
  f_1^\cst\scolon A\sto B_1 \quad f_2^\modi\scolon  A \sto B_2 \squad 
  \pr_1^\pure \circ g \eqw f_1 \squad \pr_2^\pure \circ g \eqs f_2}
  {g^\modi \eqs \lpair{f_1,f_2}^\modi}    \\
\hline
\rul{effect} &
\dfrac{f,g\colon A \to B \quad f\eqw g \quad 
  \pa_A\circ f \eqs \pa_A\circ g}
  {f\eqs g} \\ 
\hline
\end{array} $$
\renewcommand{\arraystretch}{1}
\caption{$\Log_\sta$: additional rules for products} 
\label{fig:state-prod} 
\end{figure}

For each set $\Loc$ of \emph{locations} (or identifiers),
additional grammar and rules for the logic $\Log_\sta$
are given in Figure~\ref{fig:state-lookup}.
We extend the grammar of $\Log_\comon$ 
with a type $V_X$, an accessor $\lookup_X^\acc:\unit\to V_X$
and a modifier $\update_X^\modi: V_T\to\unit$ for each location $X$,
and we also extend its rules. 

The rule \rul{local-global} asserts that two functions 
without result 
coincide as soon as they coincide when observed at each location.
Together with the rule \rul{effect} it implies that two functions 
coincide as soon as they return the same value 
and coincide on each location.

\begin{figure}[htbp]   
\renewcommand{\arraystretch}{2}
$$ \begin{array}{|ll|} 
\hline
\textrm{Types: } & t::=\;  V_X  \;\; \mbox{ for each } X\in\Loc \\ 
\textrm{Terms: } & f::=\;  \lookup_X \mid \update_X  
   \;\; \mbox{ for each } X\in\Loc  \\ 
\hline
\rul{lookup} & 
\dfrac{X\in\Loc}{\lookup_X^\acc\colon \unit\to V_X}  \\ 
\rul{update} & 
\dfrac{X\in\Loc}{\update_X^\modi\colon V_X\to\unit} \\ 
\rul{lookupdate} & 
\dfrac{X\in\Loc}{\lookup_X\circ \update_X \eqw \id_{V_X}} 
\quad
\dfrac{X,Y\in\Loc\squad X\ne Y}
  {\lookup_Y\circ \update_X \eqw \lookup_Y\circ \pa_{V_X} }  \\ 
\rul{local-global} & 
\dfrac{ f,g\colon A\to \unit \quad \mbox{for all }X\in \Loc\; 
  \lookup_X\circ f \eqw \lookup_X\circ g}
  {f\eqs g} \\
\hline
\end{array} $$
\renewcommand{\arraystretch}{1}
\caption{$\Log_\sta$: additional grammar and rules for states} 
\label{fig:state-lookup} 
\end{figure}

For each family of objects $(V_X)_{X\in\Loc}$ in $\cC$  
such that $S\cong \prod_{X\in\Loc}V_X$ 
we build a model $\cC_\sta$ of $\Log_\sta$,
which extends the model the model $\cC_\coM$ of $\Log_\comon$
with functions for looking up and updating the locations.

The types $V_X$ are interpreted as the objects $V_X$ and  
the accessors $\lookup_X^\acc:\unit\to V_X$ as the projections 
from $S$ to $V_X$. Then the interpretation of each modifier 
$\update_X^\modi:V_X\to \unit$ is the function from $V_X\times S$ to $S$ 
defined as the tuple 
of the functions $f_{X,Y}:V_X\times S \to V_Y$ 
where $f_{X,X}$ is the projection from $V_X\times S$ to $V_X$ 
and $f_{X,Y}$ is made of the projection from $V_X\times S$ to $S$
followed by $\lookup_Y:S\to V_Y$ when $Y\ne X$.

The logic we get, and its model, are essentially the same as in 
\cite{DDFR12-state}: thus, the pattern for a comonad 
in Section~\ref{sec:patterns} can be seen as 
a generalization to arbitrary comonads of the approach in \cite{DDFR12-state}. 

Since we have assumed that the category $\cC$ is \emph{distributive} 
we get new decorations for the rule patterns for coproducts:
the copair of two modifiers now exists, the corresponding decorated rules 
are given in Figure~\ref{fig:state-coprod}. 

The interpretation of the modifier $\copair{f_1|f_2}$,
when both $f_1$ and $f_2$ are modifiers, 
is the composition of $\copair{f_1|f_2}\colon 
(A_1\times S)+(A_2\times S)\to B\times S$ with the inverse of 
the canonical morphism $(A_1\times S)+(A_2\times S)\to (A_1+A_2)\times S$:
this inverse exists because  $\cC$ is distributive.

\begin{figure}[htbp]   
\renewcommand{\arraystretch}{2}
$$ \begin{array}{|ll|} 
\hline
\rul{copair} & 
\dfrac{ f_1^\modi\colon  A_1 \to B \quad f_2^\modi\colon  A_2 \to B}
  {\copair{f_1|f_2}^\modi\colon  A_1\!+\!A_2 \to B } \\
\rul{copair-eq} & 
\dfrac{ f_1^\modi\colon  A_1 \to B \squad f_2^\modi\colon  A_2 \to B}
  {\copair{f_1|f_2} \circ \copr_1 \eqs  f_1 \quad 
   \copair{f_1|f_2} \circ \copr_2 \eqs  f_2  } \\
\rul{copair-u} & 
\dfrac{f_i^\modi\scolon  A_i \sto B \squad 
  g^\modi\scolon  A_1\!\splus\!A_2 \sto B \squad 
  g\circ \copr_i \eqs f_i }
  {g^\modi \eqs \copair{f_1|f_2}^\modi } \\
\hline
\end{array} $$
\renewcommand{\arraystretch}{1}
\caption{$\Log_\sta$: additional rules for coproducts, 
when $\cC$ is distributive} 
\label{fig:state-coprod} 
\end{figure}

\subsection{States: conditionals and binary operations}\

To conclude with states, let us look at the 
constructions for conditionals and binary operations
in the language for states. 

The rules in Figure~\ref{fig:state-coprod} 
provide conditionals.

There is no binary product of modifiers, 
but there is a left product of a constructor and a modifier
and a right product of a modifier and a constructor.
It follows that the \emph{left and right sequential products} 
of two modifiers $f_1$ and $f_2$ 
can be defined, as in \cite{DDR11-seqprod},
by composing, e.g., the left product of an identity and $f_1$ 
with the right product of $f_2$ and an identity.

A major feature of this approach is that, for states,  
sequential products are defined without any new ingredient: 
no kind of strength, in contrast with the approach using 
the strong monad of states $(A\times S)^S$ \cite{Moggi91},
no ``external'' decoration for equations, in contrast with 
\cite{DDR11-seqprod}. 
This property is due to the introduction of the intermediate notion 
of \emph{accessors} 
between pure terms (or \emph{values}) and modifiers (or \emph{computations}). 

\subsection{Hilbert-Post completeness}\

Now we use the decorated logic $\Log_\sta$ for proving that 
the decorated theory for states is Hilbert-Post complete. 
This result is proved in \cite[Prop.2.40]{Pretnar10} 
in the framework of Lawvere theories. 
Here we give a proof in the decorated logic for states. 
This proof has been checked in Coq\footnote{Effect categories and COQ,
\url{http://coqeffects.forge.imag.fr}}.

The logic we use is the fragment $\Log_{\sta,0}$ of $\Log_\sta$ which 
involves neither products nor coproducts nor the empty type
(but which involves the unit type). 
The \emph{theory of state}, denoted $\Tt_\sta$, 
is the family of equations which may be derived 
from the axioms of $\Log_{\sta,0}$ using the rules of $\Log_{\sta,0}$.
More generally, a \emph{theory} $\Tt$ with respect to $\Log_{\sta,0}$ 
is a family of equations between terms of $\Log_{\sta,0}$ 
which is saturated with respect to the rules of $\Log_{\sta,0}$. 
A theory $\Tt'$ is an \emph{extension} of a theory $\Tt$ 
if it contains all the equations of $\Tt$. 
Two families of equations are called \emph{equivalent} if 
each one can be derived from the other with the rules of $\Log_{\sta,0}$.

As in \cite[Prop.2.40]{Pretnar10}, 
for the sake of simplicity it is assumed
that there is a single location $X$,
and we write $V$, $\lookup$ and $\update$ instead of 
$V_X$, $\lookup_X$ and $\update_X$.
Then there is a single axiom 
$  \lookup \circ \update \eqw \id_V $.

In addition, it is assumed that all types are \emph{inhabited}, 
in the sense that 
for each type $X$ there exists a closed pure term with type $X$.

\begin{theorem}\label{theo:equations}
Every equation between terms of $\Log_{\sta,0}$ 
is equivalent to four equations between pure terms.
\end{theorem}

\begin{proof}
The proof is obtained by merging the two parts of 
Proposition~\ref{prop:equations}, 
which is proved in Appendix~\ref{app-complete}.
\end{proof}

\begin{proposition}\
\label{prop:equations}
\begin{enumerate}
\item \label{prop:equations-acc} 
Every equation between accessors is equivalent 
to two equations between pure terms. 
\item \label{prop:equations-modi}
Every equation between modifiers is equivalent 
to two equations between accessors. 
\end{enumerate}
\end{proposition}

Roughly speaking, 
a theory (with respect to some logic) is said \emph{syntactically complete}
if no unprovable axiom can be added to the theory 
without introducing an inconsistency. 
More precisely, 
a theory with respect to the equational logic  
is \emph{Hilbert-Post complete} if it is consistent and has no 
consistent proper extension \cite[Definition~2.8.]{Pretnar10}.
Since we use a decorated version of the equational logic, 
we have to define a decorated version of Hilbert-Post completeness.

\begin{definition}
\label{defi:complete}
With respect to the logic $\Log_{\sta,0}$,  

a theory $\Tt$ is \emph{consistent} if there is an equation 
which is not in $\Tt$. 

An extension $\Tt'$ of a theory $\Tt$ is a \emph{pure extension} 
if it is generated by $\Tt$ and by equations between pure terms.
It is a \emph{proper extension} if it is not a pure extension.

A theory $\Tt$ is \emph{Hilbert-Post complete} if it is consistent and 
has no consistent proper extension. 
\end{definition}

The proof of Theorem~\ref{theo:complete} 
relies on Theorem~\ref{theo:equations}. 
We do not have to assume that the interpretation of the type $V$
is a countable set.
We have assumed that $\Loc$ is a singleton, but 
we conjecture that our result can be generalized to any 
set of locations, without any finiteness condition. 

\begin{theorem}
\label{theo:complete}
The theory for state is Hilbert-Post complete.
\end{theorem}

\begin{proof} 
The theory $\Tt_\sta$ is consistent:
it cannot be proved that $\update^\modi \eqs \pa_V^\pure$.

Let us consider an extension $\Tt$ of $\Tt_\sta$
and let $\Tt_\pure$ be the theory generated by $\Tt_\sta$ and 
by the equations between pure terms in $\Tt$.
Thus, $\Tt_\pure$ is a pure extension of $\Tt_\sta$ 
and $\Tt$ is an extension of $\Tt_\pure$.
Let us consider an arbitrary equation $e$ in $\Tt$,
according to Theorem~\ref{theo:equations}
we get a family $E$ of equations between pure terms 
which is equivalent to the given equation $e$.

Since $e$ is in $\Tt$ and $\Tt$ is saturated, 
the equations in $E$ are also in $\Tt$, hence they are in $\Tt_\pure$.

Since $E$ is in $\Tt_\pure$ and $\Tt_\pure$ is saturated, 
the equation $e$ is also in $\Tt_\pure$.

This proves that $\Tt_\pure=\Tt$,
so that the theory $\Tt_\sta$ has no proper extension. 
\end{proof}

\section{Exceptions: an instance of the pattern for monads}
\label{sec:exceptions}

\subsection{The core language for exceptions}\

Let us consider a bicartesian category $\cC$ 
with monomorphic coprojections 
and with a distinguished object $E$ called the \emph{object of exceptions}.
We do not assume that $\cC$ is distributive (it would not help)
nor codistributive, because usually this is not the case.
The \emph{monad of exceptions} on $\cC$ is the monad $(\M,\eta,\mu)$  
with endofunctor $\M A=A+E$, its unit $\eta$ is made of the coprojections 
$\eta_A\colon A\to A+E$, and its multiplication $\mu$ ``merges'' the exceptions,
in the sense that $\mu_A=\copair{\id_{A+E}|\copr_A}\colon (A+E)+E\to A+E$ where 
$\copr_A\colon E\to A+E$ is the coprojection. 
It satisfies the mono requirement because 
the coprojections are monomorphisms.
Thus, the category $\cC$ with the monad of exceptions 
provides a model of the logic $\Log_\mon$. 
The name of the decorations is adapted to the monad of exceptions: 
a constructor is called a \emph{propagator}: 
it may raise an exception but cannot recover from an exception, 
so that it has to propagate all exceptions;
a modifier is called a \emph{catcher}.

For this specific monad $\M A=A+E$, 
it is possible to extend the logic $\Log_\mon$ as $\Log_\exc$,
called the \emph{logic for exceptions}, 
so that $\cC$ with $\M A=A+E$ can be extended as a model 
$\cC_\exc$ of $\Log_\exc$.

First, dually to the left and right pairs for states in 
Figure~\ref{fig:state-prod}, we get 
new decorations to the rule patterns for copairs in $\Log_\mon$, 
involving modifiers, as in Figure~\ref{fig:exc-coprod}
for the left copairs (the rules for the right copairs are omitted).

The interpretation of the left copair $\lcopair{f_1|f_2}^\modi:
A_1+A_2\to B$ is the copair $\copair{f_1|f_2}:A_1+A_2+E\to B+E$ 
of $f_1:A_1\to B+E$ and $f_2:A_2+E\to B+E$ in~$\cC$.

For instance, the coproduct of $A \cong A+\empt$, 
with coprojections $\id_A^\pure:A\to A$ and $\copa_A^\pure:\empt\to A$, 
gives rise to the left copair $\lcopair{f_1|f_2}^\modi:A \to B$ 
of any constructor $f_1^\cst\colon  A \to B$
with any modifier $f_2^\modi\colon \empt \to B$,
which is characterized up to strong equations
by $\lcopair{f_1|f_2}\eqw f_1$ and $\lcopair{f_1|f_2}\eqs f_2$.
This will be used in the construction of the $\try/\catch$ expressions. 

Moreover, the rule \rul{effect} expresses the fact that,
when $\M A=A+E$, two modifiers coincide as soon as they coincide 
on ordinary values and on exceptions. 

\begin{figure}[htbp]   
\renewcommand{\arraystretch}{2}
$$ \begin{array}{|ll|} 
\hline
\rul{l-copair} &
\dfrac{ f_1^\cst\colon  A_1 \to B \quad f_2^\modi\colon  A_2 \to B}
  {\lcopair{f_1|f_2}^\modi\colon  A_1\!+\!A_2 \to B }  \\
\rul{l-copair-eq} &
\dfrac{ f_1^\cst\colon  A_1 \to B \squad f_2^\modi\colon  A_2 \to B}
  {\lcopair{f_1|f_2}^\modi \circ \copr_1^\pure \eqw  f_1^\cst \quad 
   \lcopair{f_1|f_2}^\modi \circ \copr_2^\pure \eqs  f_2^\modi  } \\
\rul{l-copair-u} &
\dfrac{g^\modi\scolon  A_1\!+\!A_2 \sto B \squad 
  f_1^\cst\scolon  A_1 \sto B \squad f_2^\modi\scolon A_2 \sto B \squad 
  g\circ \copr_1 \eqw f_1 \squad g\circ \copr_2 \eqs f_2}
  {g^\modi \eqs \lcopair{f_1|f_2}^\modi}    \\
\hline
\rul{effect} &
\dfrac{f,g\colon A \to B \quad f\eqw g \quad 
  f\circ \copa_A \eqs g\circ \copa_A}
  {f\eqs g}  \\ 
\hline
\end{array} $$
\renewcommand{\arraystretch}{1}
\caption{$\Log_\exc$: additional rules for coproducts} 
\label{fig:exc-coprod} 
\end{figure}

For each set $\Ename$ of \emph{exception names}, 
additional grammar and rules for the logic $\Log_\exc$
are given in Figure~\ref{fig:exc-tag}.
We extend the grammar of $\Log_\mon$ 
with a type $V_T$, a propagator $\tagg_T^\acc:V_T\to\empt$
and a catcher $\untag_T^\modi:\empt\to V_T$ for each exception name $T$,
and we also extend its rules. 

The logic $\Log_\exc$ obtained performs the \emph{core} operations 
on exceptions: 
the \emph{tagging} operations 
encapsulate an ordinary value into an exception, 
and the \emph{untagging} operations 
recover the ordinary value which has been encapsulated in an exception.

This may be generalized by assuming a hierarchy of exception names 
\cite{DDR13-exc}. 

The rule \rul{local-global} asserts that two functions without argument 
coincide as soon as they coincide on each exception.  
Together with the rule \rul{effect} it implies that two functions 
coincide as soon as they coincide on their argument and on each exception.

\begin{figure}[htbp]   
\renewcommand{\arraystretch}{2}
$$ \begin{array}{|ll|} 
\hline
  \textrm{Types: } & t::=\; V_T  \;\; \mbox{ for each } T\in\Ename \\ 
  \textrm{Terms: } & f::=\; \tagg_T \mid \untag_T  
  \;\; \mbox{ for each } T\in\Ename \\ 
\hline
\rul{tag} & 
\dfrac{T\in \Ename}{\tagg_T^\cst\colon V_T\to \empt}  \\ 
\rul{untag} & 
\dfrac{T\in \Ename}{\untag_T^\modi\colon \empt\to V_T} \\ 
\rul{untag-tag} & 
\dfrac{T\in \Ename}{\untag_T\circ \tagg_T \eqw \id_{V_T}} 
\qquad
\dfrac{T,R\in \Ename\; T\ne R}
  {\untag_T\circ\tagg_R\eqw\copa_{V_T}\circ\tagg_R }  \\ 
\rul{local-global} & 
\dfrac{ f,g\colon \empt\to B \quad \mbox{for all }T\in \Ename\; 
  f\circ \tagg_T \eqw g\circ \tagg_T}
  {f\eqs g} \\
\hline
\end{array} $$
\renewcommand{\arraystretch}{1}
\caption{ $\Log_\exc$: additional grammar and rules for exceptions} 
\label{fig:exc-tag} 
\end{figure}

For each family of objects $(V_T)_{T\in\Ename}$ in $\cC$  
such that $E \cong \sum_{T\in\Ename}V_T $ 
we build a model $\cC_\exc$ of $\Log_\exc$,
which extends the model the model $\cC_\M$ of $\Log_\mon$
with functions for tagging and untagging the exceptions. 

The types $V_T$ are interpreted as the objects $V_T$ and  
the propagators $\tagg_T^\cst:V_T\to\empt$ as the coprojections 
from $V_T$ to $E$. Then the interpretation of each catcher 
$\untag_T^\modi:\empt\to V_T$ is the function $\untag_T:E\to V_T+E$ 
defined as the cotuple (or case distinction)
of the functions $f_{T,R}:V_R\to V_T+E$  
where $f_{T,T}$ is the coprojection of $V_T$ in $V_T+E$
and $f_{T,R}$ is made of $\tagg_R:V_T\to E$ followed by
the coprojection of $E$ in $V_T+E$ when $R\ne T$.

This can be illustrated, in an informal way, 
as follows: $\tagg_T$ encloses its argument $a$ in a box 
with name $T$, while $\untag_T$ opens every box with name $T$ to recover 
its argument 
and returns every box with name $R \ne T$ without opening it:
 $$ 
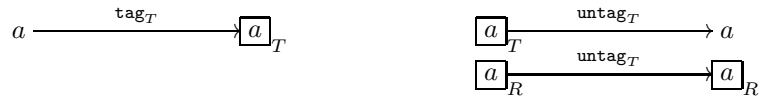
\xymatrix@C=3pc@R=.5pc{
& a \ar[rr]^{\tagg_T} && *+[F-]{a} \ar@{}[r]_(.2){T} &&
*+[F-]{a} \ar@{}[r]_(.2){T} \ar[rr]^{\untag_T} && a \\ 
&&&&& *+[F-]{a} \ar@{}[r]_(.2){R} \ar[rr]^{\untag_T} && 
  *+[F-]{a} \ar@{}[r]_(.2){R} & \\
}$$

Since we did not assume that the category $\cC$ is codistributive 
we cannot get products of modifiers 
in a way dual to the coproducts of modifiers for states.

However these rules have not been used for proving the
Hilbert-Post completeness of the theory for state. Thus 
by duality from Theorem~\ref{theo:complete} we get ``for free'' 
a result about the core language for exceptions.

\begin{corollary}
\label{coro:complete}
The core theory for exceptions is Hilbert-Post complete.
\end{corollary}

\subsection{The programmer's language for exceptions}\

We have obtained a logic $\Log_\exc$ for exceptions, 
with the core operations for tagging and untagging. 
This logic provides a direct access to 
catchers (the untagging functions), which is not provided by 
the usual mechanism of exceptions in programming languages. 
In fact the core operations remain \emph{private}, while 
there is a \emph{programmer's} language, 
which is \emph{public}, with no direct access to the catchers.

The programmer's language for exceptions
provides the operations for \emph{raising} and \emph{handling} exceptions,
which are defined in terms of the core operations. 

This language has no catcher: 
the only way to catch an exception is by using a $\try/\catch$ expression, 
which itself propagates exceptions. 
Thus, all terms of the programmer's language are propagators. 
This language does not include the private tagging and untagging operations,
but the public $\throw$ and $\try/\catch$ constructions,
which are defined in terms of $\tagg$ and $\untag$. 
For the sake of simplicity we assume that only one type of exception is
handled in a $\try/\catch$ expression,
the general case is treated in \cite{DDR13-exc}. 

The main ingredients for building the programmer's language 
from the core language are the coproducts $A\cong A+\empt$ 
and a new conversion rule for terms. 
The \emph{downcast} conversion of a catcher to a propagator 
could have been defined in Section~\ref{sec:patterns}
for the logic $\Log_\comon$,
and dually for the logic $\Log_\mon$; the rule is:  
  $$ \dfrac{f^\modi\colon A\to B}{(\downcast f)^\cst\colon A\to B} $$
This downcasting conversion from catchers to propagators   
is interpreted by mapping $f\colon \M A\to \M B$ 
to $\downcast f=f\circ\eta_A\colon A\to \M B$.
It is related to weak equations: 
$f \eqw \downcast f$, and 
$f \eqw g$  if and only if $ \downcast f\eqs \downcast g$. 
But the downcasting conversion is \emph{unsafe}:
several catchers may be downcasted to the same propagator. 
This powerful operation turns an effectful term to an effect-free one;  
since it is not required for states nor for the core language for exceptions, 
we did not introduce it earlier.

\begin{definition}
\label{defi:exc}
For each type $B$ and each exception name $T$, the propagator
$ \throw_{B,T}^\cst$ is: 
 $$ \throw_{B,T}^\cst = \copa_B^\pure \circ \tagg_T^\cst \colon V_T\to B $$
For each each propagator $f^\cst\colon A\to B$, each exception name $T$ and 
each propagator $g^\cst\colon V_T\to B$, 
the propagator $\try(f)\catch(T\To g)^\cst$ is defined in three steps, 
involving two catchers 
$\catch(T\To g)^\modi $ and $\TRY(f)\catch(T\To g)^\modi$,
as follows: 
\renewcommand{\arraystretch}{1.3}
$$ \begin{array}{l}
\catch(T\To g)^\modi = 
  \copair{\; g^\cst \;|\; \copa_B^\pure \;}^\cst \circ\untag_T^\modi 
  \colon  \empt\to B \\ 
\TRY(f)\catch(T\To g)^\modi = 
  \lcopair{\; \id_B \;|\; \catch(T\To g) \;}^\modi \circ f^\cst 
  \colon A\to B \\
\try(f)\catch(T\To g)^\cst = 
  \downcast(\TRY(f)\catch(T\To g)) 
  \colon A\to B \\
\end{array}$$
\renewcommand{\arraystretch}{1}
\end{definition}

This means that raising an exception with name $T$
consists in tagging the given ordinary value (in $V_T$) 
as an exception and coerce it to any given type $B$. 

For handling an exception, 
the intermediate expressions $\catch(T\To g)$ 
and $\TRY(f)\catch(T\To g)$ are private catchers 
and the expression $\try(f)\catch(T\To g)$ is a public propagator: 
the downcast operator prevents it from catching exceptions 
with name $T$ which might have been raised before 
the $\try(f)\catch(T\To g)$ expression is considered. 

The definition of $\try(f)\catch(T\To g)$ corresponds to 
the Java mechanims for exceptions \cite{java,Jacobs01}.

The definition of $\try(f)\catch(T\To g)$ corresponds to the following 
control flow,
where \texttt{exc?} means ``\emph{is this value an exception?}'',
an \emph{abrupt} termination returns an uncaught exception 
and a \emph{normal} termination returns an ordinary value;
this corresponds, for instance, to the Java mechanims for exceptions
\cite{java,Jacobs01}.

$$ \xymatrix@C=.8pc@R=.5pc{
& \ar[d] && \\
& \isexc \ar[ld]_{Y}\ar[rd]^{N} && \\
\abr && f^\cst \ar[d] & \\
&& \isexc \ar[ld]_{Y}\ar[rd]^{N} & \\
& \untag_T^\modi \ar[d] && \nor \\
& \isexc \ar[ld]_{Y}\ar[rd]^{N} && \\ 
\abr && g^\cst \ar[d] & \\
&& \txt{\nor \mbox{ or } \abr} \\
} $$

\subsection{Exceptions: case distinction and binary operations}\

To conclude with exceptions, let us look at the 
constructions for case distinction and binary operations
in the programmer's language for exceptions,
which means, copairs and pairs of \emph{constructors}. 

The general rules of the logic $\Log_\mon$ include 
coproducts of constructors (Figure~\ref{fig:pattern-rules}), 
which provide case distinction for all terms in the programmer's 
language for exceptions.

But the general rules for a monad do not include binary products 
involving a constructor,
hence they cannot be used for dealing with binary operations 
in the programmer's language for exceptions 
when at least an argument is not pure.
Indeed, if $f_1^\pure\colon A\to B_1$ is pure 
and $f_2^\cst\colon A\to B_2$ does raise an exception, 
it is in general impossible to find $f^\cst\colon A\to B_1\times B_2$ 
such that $\pr_1\circ f\eqs f_1$ and $\pr_2\circ f\eqs f_2$. 

However, there are several ways to formalize the fact of 
first evaluating $f_1$ then $f_2$:
for instance by using a strong monad \cite{Moggi91},  
or a sequential product \cite{DDR11-seqprod},
or productors \cite{Tate13}. 
The sequential product approach can be used in our framework; 
it requires the introduction of a third kind of ``equations'',
in addition to the strong and weak equations, 
which corresponds to the usual order between partial functions:
details are provided in \cite{DDR11-seqprod}.

\section{Conclusion}

We have presented two patterns giving sound inference systems for 
effects arising from a monad or a comonad.

We also gave detailed examples of applications of these patterns to
the state and the exceptions effects.
The obtained decorated proof system for states has been implemented in Coq, 
so that the given proofs can be automatically verified. 
We plan to adapt this logic to local states (with allocation) 
in order to provide a decorated proof of the completeness Theorem 
in \cite{Staton10-fossacs}. 

From this implementation, we plan to extract the generic
part corresponding to the comonad pattern, dualize it and extend it to
handle the programmer's language for exceptions.

Then a major issue is scalability: how can we combine effects? 
Within the framework of this paper, it may seem difficult to guess
how several effects arising from either monads or comonads
can be combined.
However, as mentioned in the Introduction,
this paper deals with two patterns for instanciating the more
general framework of decorated logics \cite{DD10-dialog}.
Decorated logics are based on spans in a relevant category of logics,
so that the combination of effects can be based on the
well-known composition of spans.

\paragraph{Acknowledgment.}
We are grateful to Samuel Mimram for enlightning discussions.

\appendix

\section{Proof of Hilbert Post completeness}
\label{app-complete}

\renewcommand{\lookup}{\mathtt{lkp}}
\renewcommand{\update}{\mathtt{upd}}

The logic used in this Appendix is the fragment $\Log_{\sta,0}$ 
of the decorated logic for states $\Log_\sta$ which 
involves neither products nor coproducts nor the empty type,
but which involves the unit type. 

For the sake of simplicity it is assumed
that there is a single location $X$,
and we write $V$, $\lookup$ and $\update$ instead of 
$V_X$, $\mathtt{lookup}_X$ and $\mathtt{update}_X$.
Then there is a single axiom
$  \lookup \circ \update \eqw \id_V $.

In Section~\ref{sec:states},
the proof of Hilbert-Post completeness in Theorem~\ref{theo:complete} 
relies on Proposition~\ref{prop:equations},
which is restated here as Proposition~\ref{prop:equations-app}.
The aim of this Appendix is to prove Proposition~\ref{prop:equations-app}.

\begin{lemma}
\label{lemm:eqns} 
The following rules can be derived: 
\begin{enumerate}
\item \label{lemm:eqns-rule-unit} 
  $ \dfrac{f^\modi,g^\modi:X\to \unit \quad 
  \lookup \circ f \eqw \lookup \circ g }{f \eqs g} $
\item \label{lemm:eqns-rule-val} 
  $ \dfrac{f^\modi,g^\modi:X\to V \quad 
  f \eqw g }{\update \circ f \eqs \update \circ g} $
\item \label{lemm:eqns-strong} 

  $ \dfrac{}{\update\circ\lookup \eqs \id_{\unit}} $
\item\label{lem:pure_mid_id}
$\dfrac{a^\acc:X\to V\quad u^\pure:V\to Y}
  {u^\pure\circ\lookup^\acc\circ\update^\modi\circ a^\acc 
   \eqw u^\pure\circ a^\acc}$
\item\label{lemm:onepureacc}
$\dfrac{x^\pure:\unit\to X}{x^\pure \eqs x^\pure\circ \pa_V^\pure\circ \lookup^\acc}$
\item\label{lemm:lkppureequal}
$\dfrac{u^\pure,w^\pure:V\to X\quad w^\pure \circ \lookup^\acc  \eqs u^\pure \circ \lookup^\acc}{w^\pure \eqs u^\pure}$ 
\item\label{lemm:lkpuseless}
$\dfrac{x^\pure:\unit\to X \quad w^\pure:V\to X \quad w^\pure \circ \lookup^\acc  \eqs x^\pure}{w^\pure \eqs x^\pure \circ \pa_V^\pure}$
\end{enumerate}
\end{lemma}

\begin{proof} 

\begin{enumerate}
\item Consequence of the observational Rule~\rul{local-global} with only one
  location. 
\item Consequence of~\ref{lemm:eqns-rule-unit}
applied to $\update \circ f, \update \circ g: X\to\unit$: 
indeed, from the axiom $ \lookup \circ \update \eqw \id_V $ we get 
$\lookup \circ\update \circ f \eqw \lookup \circ\update \circ g$.  
\item 
From axiom $ \lookup \circ \update \eqw \id_V $
by substitution we get 
$\lookup \circ \update \circ \lookup \eqw \lookup $;
thus, point~\ref{lemm:eqns-rule-unit}   
implies $\update\circ\lookup \eqs \id_\unit$.
\item 
From $\lookup^\acc\circ\update^\modi\eqw \id_V$, as $u$ is pure, by the weak
replacement we have  
$u^\pure\circ\lookup^\acc\circ\update^\modi\eqw u^\pure$.
Then, weak substitution with $a$ yields
$u^\pure\circ\lookup^\acc\circ\update^\modi\circ a^\acc\eqw u^\pure\circ
a^\acc$.
\item We know that $ \pa_V^\pure \circ \lookup^\acc:\unit\to\unit \eqs \id_\unit $.

It follows that 
$x^\pure\circ \pa_V^\pure \circ \lookup^\acc \eqs x^\pure $.
\item 
Let $ w^\pure \circ \lookup^\acc  \eqs u^\pure \circ \lookup^\acc $.
Composing with $\update$ we get 
$ w^\pure \circ \lookup^\acc \circ \update^\modi \eqs 
u^\pure \circ \lookup^\acc \circ \update^\modi $. 
Using the axiom $ \lookup \circ \update \eqw \id_V $ and the replacement rule 
for $\eqw$, which can be used here because both $ w$ and $u$ are pure, 
we get  $ w^\pure \eqw u^\pure$. Since weak and strong equations coincide 
on pure terms we get $ w^\pure \eqs u^\pure$. 
\item 
Let $ w^\pure \circ \lookup^\acc  \eqs x^\pure $. 
By point~\ref{lemm:onepureacc} above we get 
$ x^\pure \eqs x^\pure\circ \pa_V^\pure \circ \lookup^\acc$, thus 
$ w^\pure \circ \lookup^\acc  \eqs x^\pure\circ \pa_V^\pure \circ \lookup^\acc$.
Then by point~\ref{lemm:lkppureequal} above we get 
$ w^\pure \eqs x^\pure \circ \pa_V^\pure$. 
\end{enumerate}
\end{proof}

Now, let us prove Proposition~\ref{prop:canonical-form}, 
which says that, up to strong equations, it can be assumed that 
there is at most one occurrence of $\lookup$ in any accessor 
and at most one occurrence of $\update$ in any modifier.

\begin{proposition} 
\label{prop:canonical-form} 
\begin{enumerate}
\item
For each accessor $a^\acc:X\to Y$, if $a$ is not pure then 
there is a pure term $v^\pure:V\to Y$ 
such that 
\begin{equation}\label{prop:canonical-form-acc} 
 a^\acc \eqs v^\pure\circ \lookup^\acc\circ \pa_X^\pure
\end{equation}
\item
For each modifier $f^\modi:X\to Y$, if $f$ is not an accessor then 
there is an accessor $a^\acc:X\to V$ and a pure term $u^\pure:V\to Y$ 
such that 
\begin{equation}\label{prop:canonical-form-modi} 
 f^\modi \eqs u^\pure\circ \lookup^\acc\circ \update^\modi \circ a^\acc
\end{equation}
\end{enumerate}
\end{proposition}

\begin{proof}
\begin{enumerate}
\item If $a^\acc:X\to Y$ is not pure then it contains 
at least one occurrence of $\lookup^\acc$. 
Thus, it can be written in a unique way as 
$ a^\acc = v^\pure\circ \lookup^\acc\circ a_1^\acc$ for some pure term 
$v^\pure:V\to Y$ and some accessor $a_1^\acc:X\to \unit$.
Since $a_1^\acc:X\to \unit$ is such that $a_1^\acc\eqs\pa_X$, 
the result follows.
\item If $f^\modi:X\to Y$ is not an accessor then it contains 
at least one occurrence of $\update^\modi$. 
Thus, it can be written in a unique way as  
$f^\modi = b^\acc \circ \update^\modi \circ f_1^\modi$ for some accessor 
$b^\acc:\unit\to Y$ and some modifier $f_1^\modi:X\to V$. From
point~\ref{prop:canonical-form-acc}, we also have that
$b^\acc \eqs v^\pure\circ \lookup^\acc\circ\pa_\unit\eqs v^\pure\circ
\lookup^\acc$  for some pure term 
$v^\pure:V\to Y$ so that $f^\modi\eqs v^\pure\circ \lookup^\acc\circ\update^\modi \circ f_1^\modi$.
  \begin{itemize}
  \item If $f_1$ is an accessor, the result follows with $a=f_1$.
  \item Otherwise, $f_1^\modi$ contains at least one occurrence 
  of $\update^\modi$. Thus, it can be written in a unique way as  
  $f_1^\modi = b_1^\acc \circ \update^\modi \circ f_2^\modi$ for some accessor 
  $b_1^\acc:\unit\to V$ and some modifier $f_2^\modi:X\to V$.
  According to point~\ref{prop:canonical-form-acc} applied to 
  the accessor $b_1$, either $b_1$ is pure or 
  $ b_1^\acc \eqs v_1^\pure\circ \lookup^\acc $ for some pure term $v_1^\pure:V\to V$
    \begin{itemize}
    \item If $ b_1^\acc \eqs v_1^\pure\circ \lookup^\acc $ then 
    $f_1 \eqs v_1 \circ \lookup \circ \update \circ f_2$. 
    The axiom $\lookup \circ \update \eqw \id_V$ and 
    the replacement and substitution rules for $\eqw$ 
    (since $v_1$ is pure) yield $f_1 \eqw v_1 \circ f_2$.
    Then it follows from point~\ref{lemm:eqns-rule-val} 
    in Lemma~\ref{lemm:eqns} that $\update \circ f_1 \eqs 
    \update \circ v_1 \circ f_2$,
    and since $f = b \circ \update \circ f_1$ we get 
    $f \eqs b \circ \update \circ v_1 \circ f_2$.
    The result follows by induction on the number of occurrences 
    of $\update$ in $f$: indeed, 
    there is one less occurrence of $\update$ in 
    $b \circ \update \circ v_1 \circ f_2$ 
    than in $f=b \circ \update \circ b_1 \circ \update \circ f_2$.
    \item If $b_1$ is pure then $b_1^\pure\eqs b_1^\pure\circ\pa_V\circ\lookup$
      from point~\ref{lemm:onepureacc} in Lemma~\ref{lemm:eqns}. 
      Thus the previous proof applies by replacing $b_1$ with $b_1\circ\pa_V$.
    \end{itemize}
  \end{itemize}
\end{enumerate}
\end{proof}

\begin{corollary} The previous forms can be simplified for accessors with
  domain~$\unit$ and for modifiers with codomain~$\unit$, as follows:
\begin{enumerate}
\item For each accessor $a^\acc:\unit\to Y$ 
there is a pure term $v^\pure:V\to Y$ 
such that $$ a^\acc \eqs v^\pure\circ \lookup^\acc $$
\item For each modifier $f^\modi:X\to \unit$  
there is an accessor $a^\acc:X\to V$ such that 
  $$ f^\modi \eqs \update^\modi \circ a^\acc $$
\end{enumerate}
\end{corollary}

\begin{proof}
\begin{enumerate}
\item \begin{itemize}
  \item If $a:\unit\to Y$ is pure, 
  since $\pa_V \circ lookup \eqs \id_\unit$ 
  (because $\pa_V \circ lookup$ is an accessor) 
  we get $ a \eqs a \circ \pa_V \circ \lookup$, 
  thus the result is obtained with $v^\pure=a \circ \pa_V$.
  \item Otherwise, we have just proved that 
  $ a \eqs v^\pure\circ \lookup\circ \pa_X^\pure $
  with $X=\unit $,
  then $\pa_X\eqs \id_\unit$ and $ a^\acc \eqs v^\pure\circ \lookup$.
  \end{itemize}

\item \begin{itemize}
  \item If $f:X\to \unit$ is an accessor,
  since $\update \circ \lookup \eqs \id_\unit$  
  we get $ f \eqs\update \circ \lookup \circ f$, 
  thus the result is obtained with $a^\acc=\lookup \circ f$. 
  \item Otherwise, we have just proved that 
  $ f \eqs b^\acc\circ \update \circ a^\acc $ with $b^\acc:\unit\to\unit $,
  then $b\eqs \id_\unit$ and $ f^\modi \eqs \update \circ a^\acc$.
  \end{itemize}

\end{enumerate}
\end{proof}

\begin{corollary}
\label{coro:downcast}
For each modifier $f^\modi:X\to Y$, if $f$ is not an accessor then 
there is an accessor $a^\acc:X\to V$ and a pure term $u^\pure:V\to Y$ 
such that $f\eqw u^\pure\circ a^\acc$.
\end{corollary}
\begin{proof}
From Proposition~\ref{prop:canonical-form} we have that 
$f^\modi \eqs u^\pure\circ \lookup^\acc\circ \update^\modi \circ
a^\acc$. Using the axiom $\lookup\circ\update\eqw\id_V$ and the
replacement rule for $\eqw$, which can be used here because $u^\pure$
is pure, we get $f^\modi\eqw u^\pure\circ a^\acc$. 
\end{proof}

We can now prove Proposition~\ref{prop:equations-app} on which the
Hilbert-Post completeness theorem relies. 
This proof has been checked with the Coq
proof assistant using the system for states of~\cite{DDEP13-coq}. The
Coq library with the inference system is available there:
\url{http://coqeffects.forge.imag.fr}. The single proof of the following
proposition (roughly 16 pages in Coq) is directly available there: 
\url{http://coqeffects.forge.imag.fr/HPcompleteCoq.v}.

\begin{proposition}
\label{prop:equations-app}
Let us assume that 
for each type $X$ there exists a closed pure term $h_X^\pure:\unit\to X$.
Then: 
\begin{enumerate}
\item 
every equation between accessors is equivalent 
to one or two equations between pure terms;
\item 
every equation between modifiers is equivalent 
to one or two equations between accessors. 
\end{enumerate}
\end{proposition}

\begin{proof}

\begin{enumerate}

\item 
  
  We prove that for any accessors $ a_1^\acc, a_2^\acc : X\to Y$ there are three
  cases: 
  \begin{enumerate} 
    \item either they are both pure and $a_1 \eqs a_2$ is the required equation
      between pure terms. 
    \item either they are both accessors and it can be derived from 
      $ a_1 \eqs  a_2 $ that $v_1 \eqs v_2$ for some pure terms
      $v_1^\pure, v_2^\pure :V\to Y$. 
    \item or one of them is pure and the other one is an accessor and it can be
      derived from $ a_1 \eqs  a_2 $ that $v_1 \eqs v_2$ and $w_1\eqs w_2$ 
      for some pure terms $v_1^\pure, v_2^\pure :V\to Y$ and 
      $w_1^\pure, w_2^\pure :X \to Y$.
    \end{enumerate}
    We prove, moreover, that the converse also hold.
  \begin{enumerate}
  \item As already mentioned, if $a_1$ and $a_2$ are both pure and $a_1 \eqs
    a_2$ is the required equation between pure terms. 
  \item If neither $a_1$ nor $a_2$ is pure, then according to
    Proposition~\ref{prop:canonical-form} 
    $a_1^\acc \eqs v_1^\pure \circ \lookup\circ \pa_X^\pure$ and
    $a_2^\acc \eqs v_2^\pure \circ \lookup\circ \pa_X^\pure$ for some pure terms 
    $ v_1^\pure, v_2^\pure :V\to Y$.
    \begin{itemize}
    \item 
      Starting from the equation $a_1^\acc\eqs a_2^\acc : X\to V$ we thus get
      $v_1^\pure\circ\lookup\circ\pa_X \eqs 
      v_2^\pure\circ\lookup\circ\pa_X:X\to Y$. 
      Then, using the assumption, 
      for any function $h_X^\pure:\unit\to X$, we have that 
      $v_1^\pure\circ\lookup\circ \pa_X\circ h_X\circ\update \eqs
      v_2^\pure\circ\lookup\circ \pa_X\circ h_X\circ\update$.
      Now $\pa_X\circ h_X^\pure \eqs \id_\unit: \unit\to\unit$. 
      This, together with the axiom $\lookup \circ \update \eqw \id_{V}$
      and the replacement rule for $\eqw$ (which can be used here because
      both $v_1$ and $v_2$ are pure) yield $ v_1^\pure \eqw v_2^\pure $. 
      As the latter are both pure terms we also have 
      $v_1^\pure\eqs v_2^\pure: V\to Y$. 
    \item 
      Conversely, if $ v_1^\pure \eqs v_2^\pure : V\to Y$ then 
      $ v_1^\pure\circ\lookup^\acc\circ\pa_X^\pure\eqs
      v_2^\pure\circ\lookup^\acc\circ\pa_X: X\to Y$, which means that
      $a_1^\acc\eqs a_2^\acc : X\to Y$.
    \end{itemize}
  \item The only remaining case is w.l.o.g. if $a_1$ is pure and $a_2$ is not. 
    \begin{itemize}
    \item Then $a_2^\acc=v_2^\pure\circ\lookup^\acc\circ \pa_X^\pure$ from
      Proposition~\ref{prop:canonical-form} as previously
      and $v_1^\pure=a_1^\pure\circ h_X^\pure\circ\pa_V:V\to Y$ satisfies
      $v_1^\pure\eqs v_2^\pure$ for any assumed $h_X^\pure:\unit\to X$.
      Indeed from $a_1^\pure\eqs v_2^\pure\circ\lookup^\acc\circ \pa_X^\pure$ we
      get 
      \begin{equation}\label{eq:fullpureacc}
        a_1^\pure\circ h_X^\pure\eqs
        v_2^\pure\circ\lookup^\acc\circ\pa_X^\pure\circ h_X^\pure .
      \end{equation}
      But, on the one hand, $a_1^\pure\circ h_X^\pure:\unit\to Y$ so that
      point~\ref{lemm:onepureacc} in Lemma~\ref{lemm:eqns} gives 
      $a_1^\pure\circ h_X^\pure\eqs v_1^\pure\circ\lookup^\acc$ with 
      $v_1^\pure=a_1^\pure\circ h_X^\pure\circ\pa_V^\pure:V\to Y$.
      On the other hand, $\pa_X^\pure\circ h_X^\pure\eqs\id_\unit^\pure$ so that
      $v_2^\pure\circ\lookup^\acc\circ \pa_X^\pure\circ h_X^\pure\eqs
      v_2^\pure\circ\lookup^\acc$. 
      Thus Equation~(\ref{eq:fullpureacc}) rewrites as 
      $v_1^\pure\circ\lookup^\acc\eqs v_2^\pure\circ\lookup^\acc$ and
      point~\ref{lemm:lkpuseless} in Lemma~\ref{lemm:eqns} yields 
      \begin{equation}\label{eq:pureaccfirst}
        a_1^\pure\circ h_X^\pure\circ\pa_V^\pure=v_1^\pure\eqs v_2^\pure : V\to Y. 
      \end{equation}
      Thus now we also have $a_2^\acc \eqs 
      v_2^\pure\circ\lookup^\acc\circ \pa_X^\pure \eqs
      a_1^\pure\circ h_X^\pure\circ\pa_V^\pure\circ\lookup^\acc\circ\pa_X^\pure \eqs
      a_1^\pure\circ h_X^\pure\circ\pa_X^\pure$. From the original equation 
      $a_1^\pure\eqs a_2^\acc$ we finally get
      \begin{equation}\label{eq:pureaccsecond}
        a_1^\pure\circ h_X^\pure\circ\pa_X^\pure \eqs a_1^\pure : X \to Y.
      \end{equation}
    \item Conversely, we start from $v_2^\pure$ and $a_1^\pure$ satisfying both
      Equations~(\ref{eq:pureaccfirst}) and~(\ref{eq:pureaccsecond}).
      Then, we define $a_2^\acc=v_2^\pure\circ\lookup^\acc\circ\pa_X^\pure$
      which satisfies 
      $a_2^\acc\eqs a_1^\pure\circ
      h_X^\pure\circ\pa_V^\pure\circ\lookup^\acc\circ\pa_X^\pure$ thanks to
      Equation~(\ref{eq:pureaccfirst}). The latter is also 
      $a_2^\acc\eqs a_1^\pure\circ h_X^\pure\circ\pa_X^\pure$ which is thus
      $a_2^\acc\eqs a_1^\pure$ thanks to Equation~(\ref{eq:pureaccsecond}).
    \end{itemize}
  \end{enumerate}

\item 
The rule \rul{effect} for states means that 
two modifiers coincide as soon as they return the same result 
and modify the state in the same way. 
This means that $f_1^\modi\eqs f_2^\modi$ if and only if 
$f_1\eqw f_2$ and $\pa_A \circ f_1 \eqs \pa_A \circ f_2$.
Thanks to Corollary~\ref{coro:downcast} 
the equation $f_1\eqw f_2$ is equivalent to an equation between accessors.
It remains to prove that the equation 
$\pa_A \circ f_1 \eqs \pa_A \circ f_2$ is also equivalent to 
an equation between accessors. 

For $i\in\{1,2\}$, since $\pa_A \circ f_i\colon A \to \unit$, 
Proposition~\ref{prop:canonical-form} 
says that $\pa_A \circ f_i \eqs \update \circ a_i$ 
for some accessor $a_i\colon A \to V$.
Thus, $\pa_A \circ f_1 \eqs \pa_A \circ f_2$
if and only if $\update \circ a_1 \eqs \update \circ a_2$.
Let us check that this equation is equivalent to $a_1 \eqs a_2$. 

Clearly if $a_1 \eqs a_2\colon A \to V$ then 
$\update \circ a_1 \eqs \update \circ a_2$.
Conversely, if $\update \circ a_1 \eqs \update \circ a_2\colon A \to \unit$
then $\lookup \circ \update \circ a_1 \eqs \lookup \circ \update \circ a_2$ 
and since $\lookup \circ \update \eqw \id_V$ we get 
$a_1 \eqw a_2$, which is the same as $a_1 \eqs a_2$ because 
$a_1$ and $a_2$ are accessors.

Thus, $\pa_A \circ f_1 \eqs \pa_A \circ f_2$
if and only if $a_1 \eqs a_2$, as required. 

\end{enumerate}

\end{proof}

\end{document}